\documentclass{article}
\usepackage{ijcai23}

\usepackage{times}
\usepackage{soul}
\usepackage{url}
\usepackage[hidelinks]{hyperref}
\usepackage[utf8]{inputenc}
\usepackage[small]{caption}
\usepackage{graphicx}
\usepackage{amsmath}
\usepackage{amsthm}
\usepackage{booktabs}
\usepackage{enumitem}
\usepackage[linesnumbered,boxed,ruled,commentsnumbered]{algorithm2e}
\usepackage{multirow}
\usepackage{balance}
\usepackage{float}
\usepackage{textcomp}
\usepackage{xcolor}
\usepackage[margin=0pt]{subcaption}

\newtheorem{definition}{Definition}
\newcommand{\model}{FeSAIL~}
\newcommand{\modelns}{FeSAIL}
\newtheorem{example}{Example}
\newtheorem{theorem}{Theorem}
\newtheorem{lemma}{Lemma}

\title{Feature Staleness Aware Incremental Learning for CTR Prediction}

\author{
Zhikai Wang$^1$
\and
Yanyan Shen$^1$\and
Zibin Zhang$^2$\And
Kangyi Lin$^2$
\affiliations
$^1$Department of Computer Science and Engineering, Shanghai Jiao Tong University\\
$^2$Tencent
\emails
\{cloudcatcher.888,shenyy\}@sjtu.edu.cn,
\{bingoozhang,plancklin\}@tencent.com
}

\begin{document}

\maketitle

\begin{abstract}
Click-through Rate~(CTR) prediction in real-world recommender systems often deals with billions of user interactions every day. To improve the training efficiency, it is common to update the CTR prediction model incrementally using the new incremental data and a subset of historical data. 
However, the feature embeddings of a CTR prediction model often get stale when the corresponding features do not appear in current incremental data. In the next period, the model would have a performance degradation on samples containing stale features, which we call the feature staleness problem. 
To mitigate this problem, we propose a \textbf{Fe}ature \textbf{S}taleness \textbf{A}ware \textbf{I}ncremental \textbf{L}earning method for CTR prediction~(FeSAIL) which adaptively replays samples containing stale features.
We first introduce a staleness aware sampling algorithm (SAS) to sample a fixed number of stale samples with high sampling efficiency. We then introduce a staleness aware regularization mechanism (SAR) for a fine-grained control of the feature embedding updating. 
We instantiate FeSAIL with a general deep learning-based CTR prediction model and the experimental results demonstrate FeSAIL outperforms various state-of-the-art methods on four benchmark datasets.
The code can be found in https://github.com/cloudcatcher888/FeSAIL.
\end{abstract}

\section{Introduction}
Click-through Rate (CTR) prediction is to estimate the probability of a user clicking a recommended item in a specific context~\cite{ipinyou1,cheng,dien,din,bert4rec,tamic}. As real-world recommender systems often encounter billions of user interactions every day, various methods~\cite{reservior6,reservior5,imsr,ssr} have been proposed to update CTR prediction model in an incremental manner, which means that model will be fine-tuned with newly arrived data and possibly with part of historical data.

Typically, CTR prediction models~\cite{deepcrossing,deepfm,ader,asmg} involve a low-level feature embedding layer, followed by the high-level feature interaction and prediction layers. 
The parameters in the high-level layers are always updated with incremental data. However in the low-level layer, the feature embeddings will not be updated and get stale when the corresponding features do not appear in the incremental data~\cite{incctr}.  
For instance, if samples with ``personal care products" feature only exist in previous time spans and missing in later time spans, the embedding of feature ``personal care products'' will be stale and become incompatible with the parameters in the high-level layers. When new samples with ``personal care products" feature appear in incremental datasets, the model may not predict the CTR on them precisely. 
We call this problem the feature staleness problem.


\begin{figure}[h]
\vspace{-0.14in}
	\centering
	\begin{subfigure}[b]{0.45\textwidth}
     \centering
	\includegraphics[width=\textwidth,height=5.5cm]{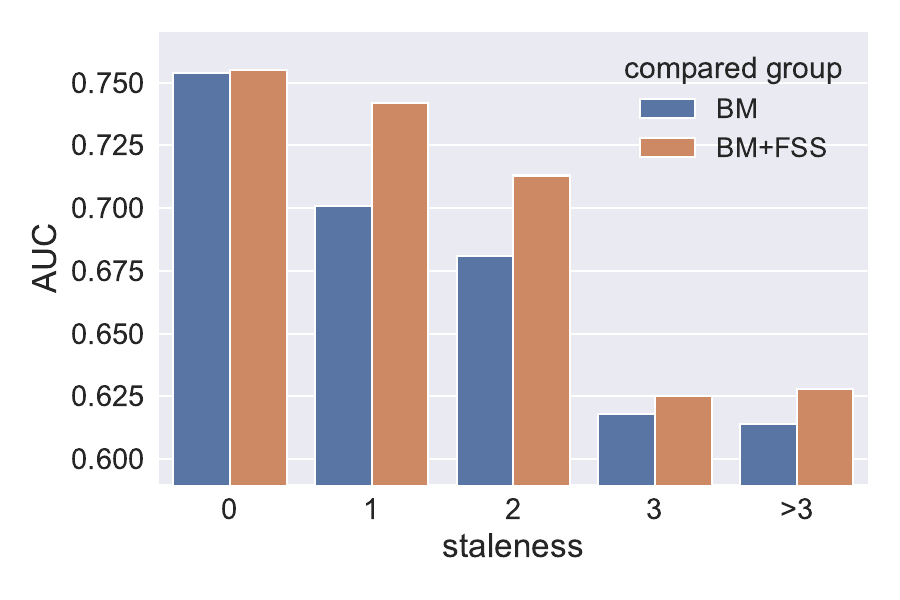}
		 \vspace{-0.24in}
	\end{subfigure}
	\caption{The observed feature staleness problem on Avazu. }
	\label{illustration}
\end{figure}

\begin{example}
We use two real CTR prediction data Criteo and Avazu~\cite{incctr} as examples which both consist of a pretraining dataset and ten consecutive incremental datasets denoted as $\{\mathcal{D}_0,\mathcal{D}_1,\cdots,\mathcal{D}_{10}\}$. 
We choose a state-of-the-art incremental learning~(IL) method~\cite{asmg} as the base incremental method~(\textbf{BM}) for illustration. In each time span $t$, we fine-tune the model parameters with $\mathcal{D}_t$. We then test the \textbf{BM} on samples in $\mathcal{D}_{t+1}$ which contain features with different \textbf{staleness}. 
The \textbf{staleness} is defined as the number of consecutive time spans for which a feature has been absent. We use the \textbf{staleness} to represent the out-of-date extent of a feature embedding. 
The average AUC performance on ten incremental datasets is reported in Figure~\ref{illustration}.
We can see that the model reports a lower average AUC on samples containing features with larger staleness, which reflects that the model suffers from the feature staleness problem on stale features, especially the feature absent for a long time, during the testing phase.
\end{example}

A vanilla way to mitigate the feature staleness problem is to replay all the samples, which contain stale features in each incremental step, to ensure all the feature embeddings are updated, which is referred as \textbf{Full Stale Sampling~(FSS)}. \textbf{FSS} can serve as a replay plug-in for the base model. As presented in Figure~\ref{illustration}, \textbf{BM+FSS} can achieve remarkable AUC improvement on samples with stale features.
However, this vanilla way has two limitations. 
First, in practice, the distribution of features can be scattered and vary among all incremental datasets. The set of all historical samples containing stale features will have an unaffordable size, and randomly sampling a subset cannot guarantee a high coverage of stale features.
Second, it treats all stale features in the same way regardless of the extent of feature's staleness, which is not reasonable. For instance, repeatedly updating low-frequency features which keep stale for a long time will increase the risk of over-fitting~\cite{dcn}. On the contrary, some features are only stale for a few time spans, which are similar to unstale features, should be updated normally. The improvement on samples with smaller staleness features is more remarkable in Figure~\ref{illustration}, which validates the second limitation.

In this paper, we propose a \textbf{Fe}ature \textbf{S}taleness \textbf{A}ware \textbf{I}ncremental \textbf{L}earning method for CTR prediction~(FeSAIL) to overcome the above mentioned limitations. 
We introduce a staleness aware sampling algorithm~(SAS) using a fixed-size reservoir to store samples with stale features. We aim to cover as many features with small staleness values as possible, which is formulated as a maximum weighted coverage problem and solved by a greedy algorithm with an approximation ratio of $1-\frac{1}{e}$ compared to the optimal solution. We then introduce a staleness aware regularization mechanism~(SAR) to restrict the updating of feature embeddings according to the extent of features' staleness.

The major contributions are summarized as follows.
\begin{itemize}[itemsep=0pt,topsep=0pt,parsep=0pt,leftmargin=10pt]
\item We observe the feature staleness problem existing in incremental learning of CTR prediction models and propose a feature staleness aware IL method named FeSAIL to mitigate this problem. 
\item We introduce the SAS algorithm based on feature staleness and solve it as a maximum weighted coverage problem greedily to guarantee the sampling efficiency. We also design the SAR algorithm for a fine-grained control of low-frequency feature updating. 
\item We instantiate FeSAIL with a general deep learning-based CTR prediction model and conduct extensive experiments, which demonstrates that FeSAIL can efficiently alleviate the feature staleness problem and achieve 1.21\% AUC improvement compared to the state-of-the-art IL methods for CTR prediction on three widely-used public datasets and a private dataset averagely.
\end{itemize}

\section{Preliminaries}
CTR prediction aims to estimate the probability that a user will click a recommended item.
Generally, CTR models include three parts: embedding layer, interaction layer and prediction layer~\cite{fpmc,deepfm,incctr,tien}.
In typical CTR prediction tasks, the input of each sample is collected in a multi-field form~\cite{incctr,dcn,fossil}. The label $y$ of each sample is 0 or 1 indicating whether a user clicked an item.
Each field $F_i~(1\le~i~\le~m)$ is filled with one specific feature $f_i$ from feature space $\mathcal{F}_i$ with size $n_i$, where $m$ is the number of fields.
Each data sample is transformed into a dense vector via an embedding layer.  
Formally, for a sample $d=\{f_1,f_2,\cdots,f_m\}$, each feature is encoded as an embedding vector $e_i~\in~\mathcal{R}^{1\times k}$, where $k$ is the embedding size. Therefore, each sample can be represented as an embedding list $E=(e_1^{\top},e_2^{\top},\cdots,e_m^{\top})\in \mathcal{R}^{m\times k}$. 
Let the total number of features be $N=\sum_{i=1}^m n_i$. The embeddings of all the features form an embedding table $\mathcal{E}\in \mathcal{R}^{N\times k}$.
In the interaction layer, existing models~\cite{deepfm,dcn,sasrec,tisasrec} utilize product operation and multi-layer perception to capture explicit and implicit feature interactions. And prediction $\hat{y}$ is generated as the probability that the user will click on a specific item. The cross-entropy loss is calculated between the label $y$ and the prediction $\hat{y}$.

Given a series of CTR prediction datasets $\{\mathcal{D}_0, \mathcal{D}_1, \cdots, \mathcal{D}_t, \cdots\}$ indexed by time span. In IL scenario, at time span $t$, $\mathcal{D}_t$ is the current dataset and $\mathcal{D}_{his}^t=\{\mathcal{D}_0,\mathcal{D}_1,\cdots,\mathcal{D}_{t-1}\}$ contains all the historical datasets. If we train the model using the current dataset and all the historical datasets, the size of the training data will grow overwhelmingly large. In many existing IL methods for CTR prediction~\cite{ader,man}, instead of preserving the whole historical datasets, they first initialize a reservoir $\mathcal{R}_0$ using $\mathcal{D}_0$, which serves as a replay buffer. As time span $t$, they generate the reservoir $\mathcal{R}_t$ based on $\mathcal{R}_{t-1} \cup \mathcal{D}_{t-1}$ and train the model using $\mathcal{R}_t$ and $\mathcal{D}_t$, which can greatly reduce the size of the training data. In this paper, we follow this setting and investigate how to find an efficient way for sampling $R_t$ from $\mathcal{R}_{t-1} \cup \mathcal{D}_{t-1}$ and update the model parameters using $\mathcal{R}_t$ and $\mathcal{D}_t$ such that the model can achieve better performance on $\mathcal{D}_{t+1}$ at each time span $t$.


%

\section{The FeSAIL Approach}
\subsection{Overview}
\begin{figure}[h]
	\centering
	\includegraphics[width=.49\textwidth,height=8.6cm]{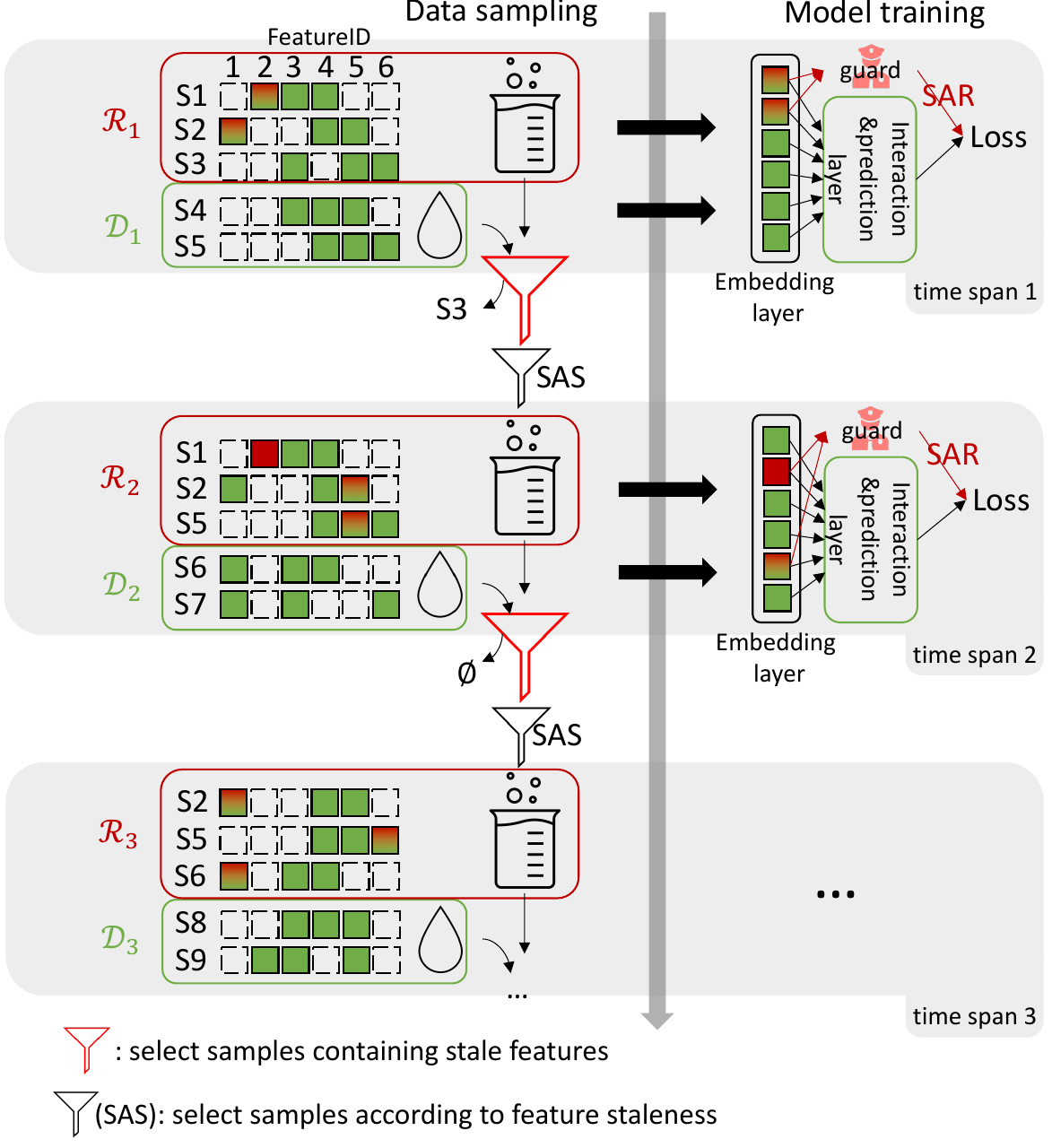}
	\vspace{-0.2in}
	\caption{The overview of \modelns. All features that appear in the current dataset are in green squares. The redness of a feature represents the extent of its staleness.}
	\label{overview}
\end{figure}
The FeSAIL is a feature staleness aware IL method for training CTR prediction models to mitigate the feature staleness problem. 
As presented in Figure~\ref{overview}, FeSAIL can be divided into sampling and training stages for each time span $t$: In the sampling stage, the staleness aware sampling~(SAS) strategy first uses a reservoir $\mathcal{R}_t$ with a fixed size to store samples containing stale features and uses a greedy algorithm to cover as much staleness features as possible. In the training stage, the model will be trained on $\mathcal{R}_t$ and the incremental dataset $\mathcal{D}_t$ with the staleness aware regularization~(SAR), which restricts the embedding updating of features according to the extent of features’ staleness.


\subsection{Staleness Aware Sampling~(SAS)}
\label{sas}
In this section, we introduce how our method chooses samples for the reservoir based on the feature staleness.
Our key insight is that resampling samples with stale features, which refer to the features not appearing in the current dataset, can help mitigate the performance degradation on new samples involving stale features, as illustrated Figure~\ref{illustration}. A na\"ive way to generate $\mathcal{R}_t$ is to choose the samples from $\mathcal{R}_{t-1} \cup \mathcal{D}_{t-1}$ containing features that do not appear in $\mathcal{D}_t$, which is referred as Reservoir-based Stale Sampling~(\textbf{RSS}). 
More formally, let $\mathcal{F}_{D_t}$ denote the feature set of $\mathcal{D}_t$. 
We add a sample $d=\{f_1,f_2,\cdots,f_m\}$ in $\mathcal{R}_{t-1} \cup \mathcal{D}_{t-1}$ to $\mathcal{R}_t$ if $d$ involves at least one feature that is not included in the feature set of $\mathcal{D}_t$,
i.e., $\exists f_i \notin \mathcal{F}_{D_t}$ where $1~\le~i\le~m$.  
The model will be trained with $\mathcal{R}_t$ and $\mathcal{D}_t$ jointly.
In this way, the chosen samples in $\mathcal{R}_t$ can supplement the feature set of $\mathcal{D}_t$ and update the feature embeddings such that the embeddings can adapt to the updated high-level layers in the model.

However, in real scenarios such as news and short video recommendation, the distribution of features can be quite scattered and vary among different incremental datasets~\cite{sml,asmg}. Even using the above mentioned method will still preserve samples with an uncontrollable size and incur high computation cost.  
A more desirable solution is to design a reservoir with a fixed size $L$ to store historical samples with stale features. 
Honestly, a fixed-sized reservoir may not be able to cover all the stale features. To mitigate the performance degeneration to the maximal extent, we compute the fixed-size reservoir based on two rules. 
First, the designed reservoir should cover as many stale features as possible. 
Second, the features with smaller staleness will be covered with higher priorities because they might have higher possibilities of reappearing again in future datasets.
We first quantify the feature staleness. Formally, each feature $f_i$ is assigned with one staleness $s_i^t$ which is initially set as 0. At time span $t$, 
the staleness of feature $f_i$ will be:
\begin{equation}
  s_i^t = 
  \begin{cases}
  s_i^{t-1}+1, & {\rm~if}~f_i~\notin~\mathcal{F}_{D_t}\\
  0, & {\rm~otherwise}
  \end{cases}
  ,
\end{equation}
where $\mathcal{F}_{D_t}$ denotes the set of feature values contained in $D_t$.
Then we define the \textit{weight} of each feature which has an inverse correlation with its \textit{staleness}:
\begin{equation}
\label{func}
w_i = func(s_i^t) + b,
\end{equation}
where $func(\cdot)$ is an inverse correlation function like an inverse proportional function or negative exponential function, and $b$ is a bias term.
Our goal is to select a fixed number of samples with as many stale features as possible, which can be formulated as follows.
\begin{definition}[Stale Features Sampling~(SFS)]
Given the reservoir $\mathcal{R}_t=\{d_1,d_2,\cdots,d_{|\mathcal{R}_t|}\}$ of \textbf{RSS} where each sample $d_l$ contains $m$ features $\{f_1,f_2,\cdots,f_m\}$ associated with weights $\{w_i\}_{i=1}^m$. The SFS problem is to find a collection of samples $\mathcal{R}_t^{fixed}\subseteq \mathcal{R}_t$, such that the total number of samples in $\mathcal{R}_t^{fixed}$ does not exceed a given capacity $L$ and the total weight of features covered by $\mathcal{R}_t^{fixed}$ is maximized.
\end{definition}

The above problem is a generalized version of the standard maximum coverage problem~(MCP)~\cite{mcp}, which has been proved that an exact solution is intractable because the search space is too large. We develop a greedy algorithm~(SAS) which can yield an approximation ratio of $1-\frac{1}{e}$ compared to the optimal solution.  

Let $W_l$ denote the total weight of the features covered by sample $d_l$, but not covered by any sample in $\mathcal{R}_t^{fixed}$. As presented in Algorithm~\ref{algsas}, SAS has $L$ iterations.
In each iteration, SAS will calculate $W_l$ for each sample in $\mathcal{R}_t$ and select the sample with the maximum $W_l$.

\begin{algorithm}
\normalsize
\caption{SAS algorithm}
\label{algsas}
\SetAlgoNoLine
\IncMargin{1em} 
\SetKwInOut{Input}{\textbf{Input}}\SetKwInOut{Output}{\textbf{Output}} 
\Input{a reservoir from \textbf{RSS} $\mathcal{R}_t=\{d_1,d_2,\cdots,d_{|\mathcal{R}_t|}\}$,\\
	 where each sample $d_l=\{f_1,f_2,\cdots,f_m\}$ \\
	  associated weights $\{w_i\}_{i=1}^m$, \\
	a given capacity $L$}
\Output{a collection of samples $\mathcal{R}_t^{fixed}\subseteq \mathcal{R}_t$}
$\mathcal{R}_t^{fixed}\leftarrow\emptyset$\; 
\For{$L$ iterations}{
		select $d_l\in \mathcal{R}_t$ that maximizes $W_l$\;
		$\mathcal{R}_t^{fixed}\leftarrow\mathcal{R}_t^{fixed}+ d_l$\;
	}
\end{algorithm}

Let $b_l$ denote the total weight till $l^{th}$ iteration, i.e., $ b_l=\sum_{j=1}^lW_j$ and $OPT$ denote the optimal solution of the SFS. Let $c_l$ denote the left weight for optimal, i.e., $c_l=OPT-b_l$. Let $b_0=0$, $c_0=OPT$. 
\begin{lemma}
In $l+1^{th}$ iteration, we always have $W_{l+1}\ge\frac{c_l}{L}$, where $W_{l+1}$ is a possible total weight.
\end{lemma}

\begin{proof}

The optimal solution reaches $OPT$ weight at $L$ iterations. That means, at each iteration $l+1$ there should be some unchosen samples in $\mathcal{R}_t$ whose contained weight greater than $\frac{c_l}{L-l}\ge\frac{c_l}{L}$. Otherwise, it was impossible to cover $OPT$ weight at $L$ steps by the optimal solution. Since the approximation algorithm is a greedy algorithm, i.e., choosing always the set covering the maximum weight of uncovered elements, the weight of chosen set at each iteration should be at least the $\frac{1}{L}$ of the weight of remaining uncovered elements.
\end{proof}

\begin{lemma}
In $l+1^{th}$ iteration, we have $c_{l+1}\le(1-\frac{1}{L})^{l+1}\cdot OPT$, where $c_{l+1}$ is the optimal left weight.
 \label{lm2}
\end{lemma}

\begin{proof}

We prove Lemma~\ref{lm2} by mathematical induction. 

\noindent \textit{Base case:} If $l=0$,

\noindent we have $OPT-c_1=b_1= W_1\ge \frac{c_1}{L}$.

\noindent So $c_1\le(1-\frac{1}{L})\cdot OPT$
and the lemma holds when $l=0$.

\noindent \textit{Inductive hypothesis:} Suppose the theorem holds for $l\ge0$.

\noindent \textit{Inductive step:} 
$c_l\le (1-\frac{1}{L})^l\cdot OPT$

\noindent So $c_{l+1}=c_l-W_{l+1}$ 
$\rightarrow c_{l+1} \le c_l -\frac{c_l}{L}$

\noindent $\rightarrow c_{l+1}\le(1-\frac{1}{L})^{l+1}\cdot OPT$.

\noindent So the theorem holds for $l+1$.
\end{proof}

\begin{theorem}
\label{approximation}
SAS can achieve an $1-\frac{1}{e}$ approximation ratio compared to the optimal solution for SFS.
\label{thrm}
\end{theorem}

\begin{proof} 

According to Lemma~\ref{lm2}, we have 

\noindent $c_L\le (1-\frac{1}{L})^L\cdot OPT=\frac{1}{e}OPT$.

\noindent so $b_L=OPT-c_L\ge OPT (1-\frac{1}{e})$.
\end{proof}

From Theorem~\ref{approximation}, we know that SAS will drop about 30\% more stale features than the optimal solution. However, the SAS chooses the features with small staleness in prior and the actual drop ratio of stale features will be about 10\%, which is testified in Section~\ref{sec:casestudy}.
The complexity of SAS is $O(|\mathcal{R}_t|*L*m)$, where $m$ is the number of feature fields and $\mathcal{R}_t|$ is the size of reservoir from \textbf{RSS} which still have high computation cost because both $|\mathcal{R}_t|$ and $L$ can be million scales. We propose neighbour-based SAS which can decrease the complexity to $O(|\mathcal{R}_t|*m+\bar{|Q|}_{\mathcal{N}}*L*m)$, where the neighbours $\mathcal{N}$ refer to the samples which have at least one common stale features with the current chosen sample and $\bar{|Q|}_{\mathcal{N}}$ denotes the mean of the neighbour numbers. Before Line 3 of Algorithm~{\ref{algsas}}, we only need to update $W_l$ of these neighbour samples rather than all samples in the reservoir. Because only samples sharing at least one stale feature will be neighbors, $\bar{|Q|}_{\mathcal{N}}$ is much smaller than the total number of samples. For example in Avazu, the $\bar{|Q|}_{\mathcal{N}}$ is in the 1e3 scale and the sample number or the total feature number is in the 1e8 scale.

\subsection{Staleness Aware Regularization~(SAR)}


Intuitively, 
the features keeping stale for a long time are probably low-frequency features. Updating their embeddings too frequently will increase the risk of over-fitting~\cite{dcn,sun2022self}. We follow the regularization-based methods~\cite{ader} to restrict the embedding updating of these kind of features using regularization term. However, a certain ratio of features only stale for one or two time spans, whose regularization term will be quite large and become an unaffordable burden for the model training.
 Thus it is unreasonable to treat all stale features equally regardless of the extent of feature staleness.
More formally, we add a regularization loss for feature $f_i$, which is called \textit{guard}:
\begin{equation}
\label{guard}
g_i = \frac{\min(s_i,\eta)}{\min(s_{max},\eta)}||\Delta e_i||_2,
\end{equation}
where $s_{max}$ is the max staleness in the current time span, $\Delta e_i$ is embedding change for feature $f_i$ between two consecutive mini-batches, and $\eta$ is a hyperparameter used to exclude the extreme staleness value. $s_{max}$ is used to normalize the staleness to mitigate the influence of the dataset's overall staleness extent. Thus the total loss of \model will be:
\begin{equation}
\label{loss}
\mathcal{L} = \sum_{d\in\mathcal{D}_t\cup\mathcal{R}_t}(\mathcal{L} _{CE,d}+\lambda\sum_{i=1}^N g_i),
\end{equation}
where $\mathcal{L} _{CE,d}$ refers to the cross entropy loss of model training on sample $d$, and $\lambda$ is the regularization coefficient. 

\subsection{Implementation Details}
\label{imple}
We implemented our FeSAIL approach using Pytorch 1.10 on a 64-bit Linus server equipped with 32 Intel Xeon@2.10GHz CPUs, 128GB memory and four RTX 2080ti GPUs. 
The proposed FeSAIL is model-agnostic. To test its effectiveness, we instantiate it on a general deep learning-based Embedding\&MLP model, a network architecture that most of the CTR prediction models developed in recent years are based on~\cite{asmg,sml,narm}. 
By default, we set the inverse correlation function in Eq.~(\ref{func}) as the inverse proportional function and the bias term as 1. 
The sampling reservoir size $L$ of $\mathcal{R}_t$ is the same size as the corresponding incremental dataset  $\mathcal{D}_t$.
We use a grid search over the hidden layer size, initial learning rate and the number of cross layers. The batch size is set from 256 to 4096. The embedding size and hidden layer sizes are chosen from 32 to 1024. The $\eta$ in Eq.~(\ref{guard}) is from 5 to 10. We choose the Adam optimizer~\cite{adam} to train the model with a learning rate from 0.0001 to 0.001 and perform early stopping in the training process. 
\subsection{Complexity Analysis}
\noindent \textbf{Time Complexity.}
The computational complexity of SAS is $O(|\mathcal{R}_t|*m+\bar{|Q|}*L*m)$ using neighbor based optimization.
The regularization term for SAR has complexity as $O((|\mathcal{R}_t|+|\mathcal{D}_t|)*m)$.
Thus the total external time complexity will be negligible compared to model training.

\noindent \textbf{Space Complexity.}
SAS needs to store $|\mathcal{R}_t|$ samples, the size of which is fixed and can be set according to the actual memory space. And the SAR is a non-parameterized module which does not require extra storage space.

\section{Experiment}
 We conduct the experiments to evaluate the performance of our proposed method and answer the following questions :
\begin{itemize}[itemsep=0pt,topsep=0pt,parsep=0pt,leftmargin=10pt]
\item \textbf{RQ1} Can the proposed \model outperform the existing IL methods on CTR prediction?
\item \textbf{RQ2} How do the SAR and SAS algorithms affect the effectiveness of \model?
\item \textbf{RQ3} How do different inverse correlation functions and biases influence the performance of \modelns?
\item \textbf{RQ4} How does the SAS deal with features with different staleness?
\end{itemize}

\subsection{Experimental Setups}
\begin{table}[t]
\normalsize
\centering


 \begin{tabular}{cccc}
 \toprule
  Dataset&\#samples&\#fields&\#features\\
 \hline
 Criteo & 10,692,219 & 26 & 1,849,038 \\ 
 iPinYou & 21,920,191& 21 & 4,893,029 \\
 Avazu & 40,183,910 & 20 & 10,922,019 \\
 Media & 104,416,327 & 337 & 60,833,522 \\
 \bottomrule
 \end{tabular}
  \caption{The statistics of the datasets.}
   \label{tab:dataset}
\end{table}

%
%

\noindent \textbf{Datasets.}
We use three real-world datasets from Criteo\footnote{https://www.criteo.com}, iPinYou\footnote{https://www.iPinYou.com}, Avazu\footnote{https://www.kaggle.com/c/avazu-ctr-prediction} and one private industrial dataset collected from a commercial media platform.
\begin{itemize}[itemsep=0pt,topsep=0pt,parsep=0pt,leftmargin=10pt]
\item \textbf{Criteo} This dataset consists of 24 days’ consecutive traffic logs from Criteo, including 26 categorical features and the first column as a label indicating whether the ad has been clicked or not.
\item \textbf{iPinYou} This dataset is a public real-world display ad dataset with each ad display information and the corresponding user clicks feedback. The data logs are organized by different advertisers and in a row-per-record format. There are 21.92M data instances with 14.89K positive labels (click) in total. The features for each data instance are categorical. 
\item \textbf{Avazu} This dataset contains users’ click behaviours on displayed mobile ads. There are 20 feature fields, including user/device features and ad attributes. 
\item \textbf{Media} This is a real CTR prediction dataset collected from a commercial media platform. We extract 48 consecutive hours of data that contains 337 categorial feature fields.
\end{itemize}

The detailed statistics of the datasets are listed in Table~\ref{tab:dataset}. 
For a fair comparison, we follow the same data preprocessing and splitting rule for all the compared methods. We omit the staleness analysis of Criteo and iPinYou which have a similar staleness feature distribution with Avazu.



\noindent \textbf{Evaluation Metrics.}
We use the following two metrics for performance evaluation: AUC (Area Under the ROC curve) and Logloss (cross entropy), which are commonly used in CTR prediction works~\cite{ipinyou,fnn}.

\noindent \textbf{Baselines.}
We compare our proposed \model with six baselines including two existing sample-based and two model-based IL methods for training CTR prediction models.
\begin{itemize}[itemsep=0pt,topsep=0pt,parsep=0pt,leftmargin=10pt]
 \item \textbf{Incremental Update (IU)} updates the model incrementally using only the new data.
  \item \textbf{Random Sample (RS)} updates the model using the new data and the randomly sampled historical data.
  \item \textbf{RMFX}~(sample-based)~\cite{rmfx} is a matrix factorization-based model which uses an active learning sampling algorithm to replay the most informative historical samples.
  \item \textbf{GAG}~(sample-based)~\cite{gag} is a method which uses a global attributed graph neural network for streaming session-based recommendation. We regard one sample in the CTR task as a piece of session in this model.
  \item \textbf{EWC}~(model-based)~\cite{ewc} is a general IL method which prevents the model from changing the informative parameters learned in previous datasets. 
  \item \textbf{ASMG}~(model-based)~\cite{asmg} is a meta-learning-based incremental method for CTR prediction models, which uses a sequentialized meta-learner to generate model parameters in the next time span.
 \end{itemize}
\textbf{RMFX+SAR} and \textbf{GAG+SAR} add the guard term in Eq.~(\ref{guard}) on the loss function~(SAR) of original RMFX and GAG. Similarly, \textbf{EWC+SAS} and \textbf{ASMG+SAS} use our sample method SAS to enhance the performance of EWC and ASMG, which is designed to prove the effectiveness of guarding term. Note that we use the same reservoir size for our SAR and sample-based baselines.

\subsection{Experimental Results: RQ1}
\begin{table*}
\normalsize
\centering

	\setlength{\tabcolsep}{1.13mm}{
 \begin{tabular}{c|cccccccccccc}
 \toprule

  Datasets &\multicolumn{3}{c}{Criteo}&\multicolumn{3}{c}{iPinYou}&\multicolumn{3}{c}{Avazu}	&\multicolumn{3}{c}{Media}\\
\midrule
 Methods& AUC & Logloss & RI~(\%) & AUC & Logloss & RI~(\%) & AUC & Logloss & RI~(\%)		& AUC & Logloss & RI~(\%)\\ 
\midrule
									IU
&	0.7329 	&	0.5284 	&	4.43 	&	0.7528 	&	0.4866 	&	4.09 	&	0.7210 	&	0.4388 	&	3.81 &	0.6316 	&	0.2884 	&	4.46 	\\ 
RS
&	0.7363 	&	0.5198 	&	3.41 	&	0.7585 	&	0.4791 	&	2.96 	&	0.7244 	&	0.4346 	&	3.11 	&	0.6379 	&	0.2842 	&	3.26 \\ \midrule RMFX
&	0.7415 	&	0.5096 	&	2.10 	&	0.7621 	&	0.4791 	&	2.72 	&	0.7306 	&	0.4328 	&	2.48 	&	0.6346 	&	0.2844 	&	3.56\\ GAG
&	0.7459 	&	0.5084 	&	1.68 	&	0.7674 	&	0.4715 	&	1.58 	&	0.7352 	&	0.4249 	&	1.26 &	0.6401 	&	0.2759 	&	1.65	\\ RMFX+SAR
&	0.7480 	&	0.5081 	&	1.50 	&	0.7672 	&	0.4728 	&	1.73 	&	0.7344 	&	0.4251 	&	1.34 	&	0.6379 	&	0.2804 	&	2.61\\ GAG+SAR
&	0.7498 	&	0.5012 	&	0.71 	&	0.7735 	&	0.4701 	&	1.04 	&	0.7380 	&	0.4215 	&	0.68 	&	0.6451 	&	0.2729 	&	0.71\\ \midrule EWC
&	0.7391 	&	0.5124 	&	2.53 	&	0.7595 	&	0.4783 	&	2.81 	&	0.7295 	&	0.4252 	&	1.69 	&	0.6374 	&	0.2757 	&	2.53\\ ASMG
&	0.7502 	&	0.5039 	&	0.95 	&	0.7693 	&	0.4708 	&	1.38 	&	0.7394 	&	0.4233 	&	0.78 &	0.6432 	&	0.2777 	&	1.72	\\ EWC+SAS
&	0.7479 	&	0.5075 	&	1.45 	&	0.7664 	&	0.4762 	&	2.14 	&	0.7343 	&	0.4196 	&	0.69 	&	0.6398 	&	0.2758 	&	1.65\\ ASMG+SAS
&	0.7535 	&	0.5001 	&	0.35 	&	0.7734 	&	0.4666 	&	0.68 	&	0.7407 	&	0.4200 	&	0.30 	&	0.6471 	&	0.2727 	&	0.52\\ \midrule \textbf{\modelns}
&	*0.7553 	&	*0.4978 	&	-	&	*0.7788 	&	*0.4636 	&	-	&	*0.7420 	&	*0.4181 	&	-	&	*0.6503 	&	*0.2712 	&	-	\\
\bottomrule
 \end{tabular}
 }
  \caption{Average AUC \& LogLoss over 10 test periods for four datasets. The results on different metrics are given by the average scores over all the incremental datasets. Note that ``RI" indicates the relative improvement of \model over the corresponding baseline. * denotes $p<0.05$ when performing the two-tailed pairwise t-test on FeSAIL with the best baselines.}
   \label{tab:experiment_result}
\end{table*}


Table~\ref{tab:experiment_result} presents the overall performance on four datasets~(the results are all averaged over ten runs), where we have the following observations.

\begin{figure}
\vspace{-.05in}
	\centering
	\begin{subfigure}[b]{0.49\textwidth}
     \centering
	\includegraphics[width=\textwidth]{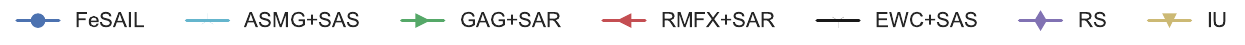}
	\vspace{-0.2in}
	\end{subfigure}
	\begin{subfigure}[b]{0.24\textwidth}
     \centering
	\includegraphics[width=\textwidth]{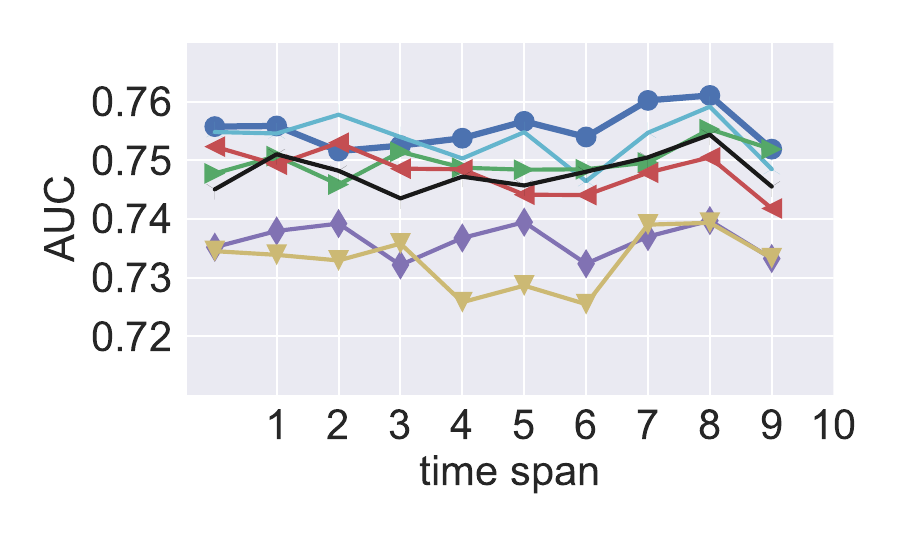}
		 \vspace{-0.24in}
	\caption{AUC on Criteo}
	\end{subfigure}
	\hspace{-3mm}
	\begin{subfigure}[b]{0.24\textwidth}
     \centering
	\includegraphics[width=\textwidth]{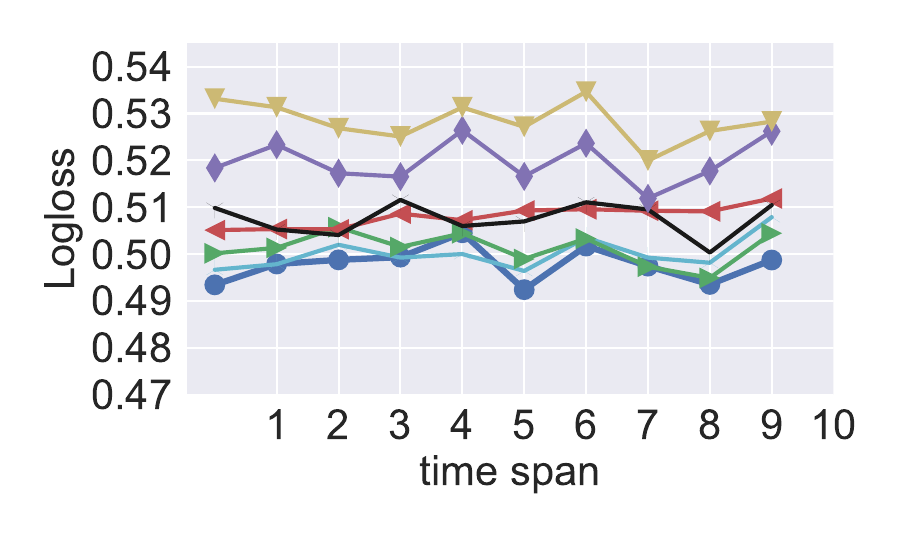}
		 \vspace{-0.24in}

	\caption{Logloss on Criteo}
	\end{subfigure}
	\begin{subfigure}[b]{0.24\textwidth}
     \centering
	\includegraphics[width=\textwidth]{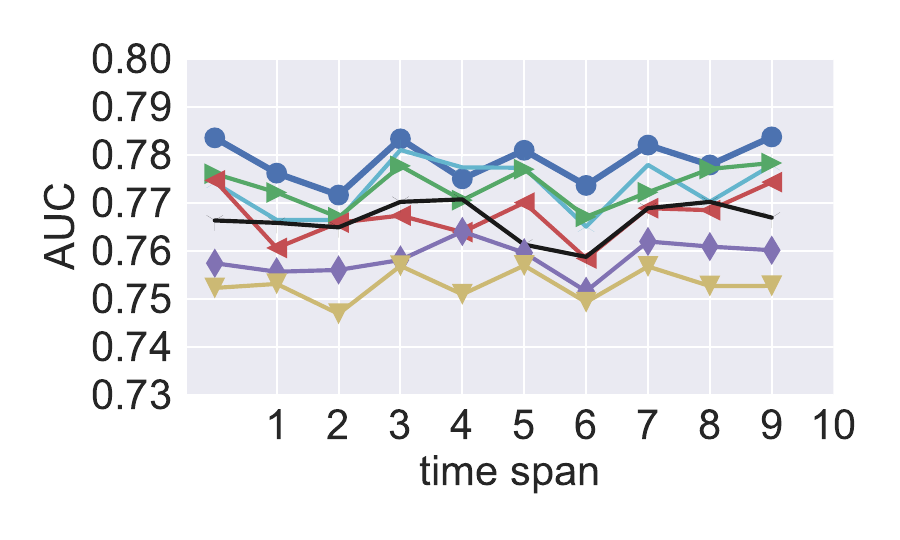}
		 \vspace{-0.24in}

	\caption{AUC on iPinYou}
	\end{subfigure}
	\hspace{-3mm}
	\begin{subfigure}[b]{0.24\textwidth}
     \centering
	\includegraphics[width=\textwidth]{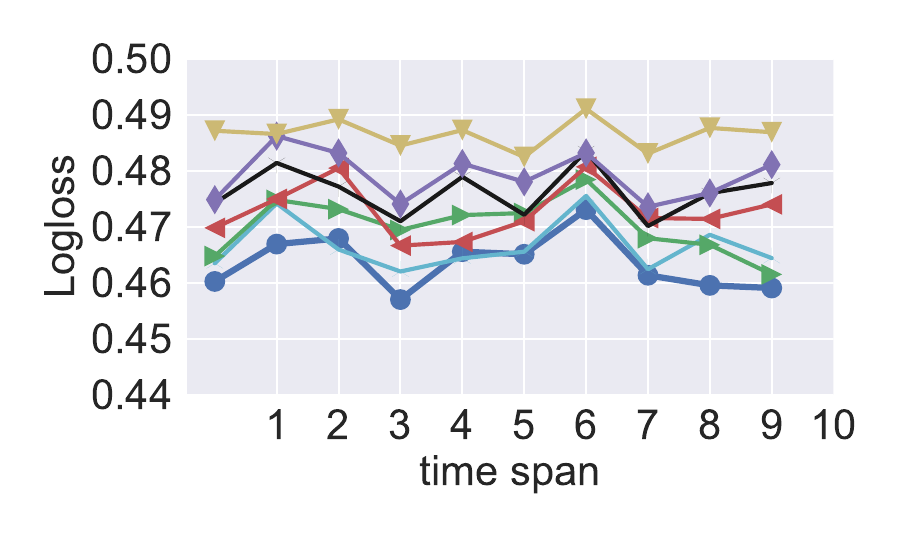}
		 \vspace{-0.24in}

	\caption{Logloss on iPinYou}
	\end{subfigure}
	
	\begin{subfigure}[b]{0.24\textwidth}
     \centering
	\includegraphics[width=\textwidth]{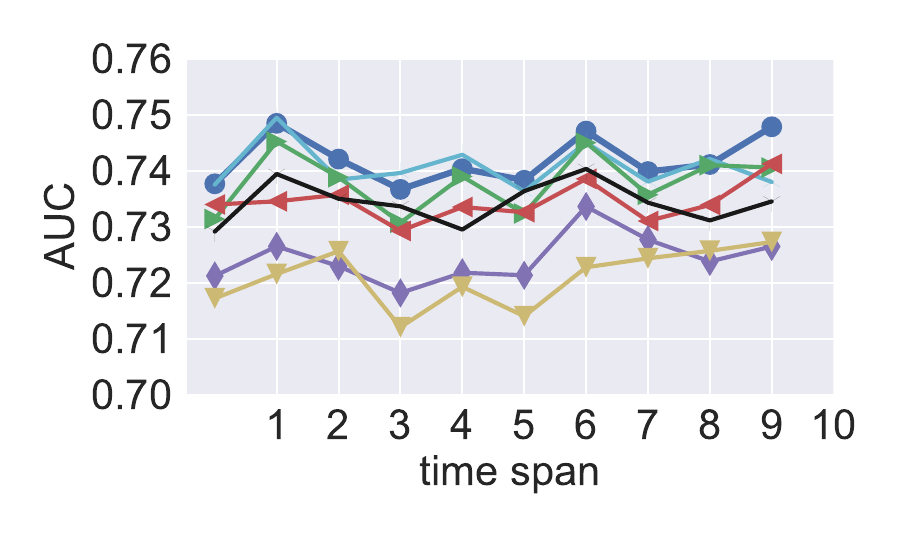}
		 \vspace{-0.26in}

	\caption{AUC on Avazu}
	\end{subfigure}
	\hspace{-3mm}
	\begin{subfigure}[b]{0.24\textwidth}
     \centering
	\includegraphics[width=\textwidth]{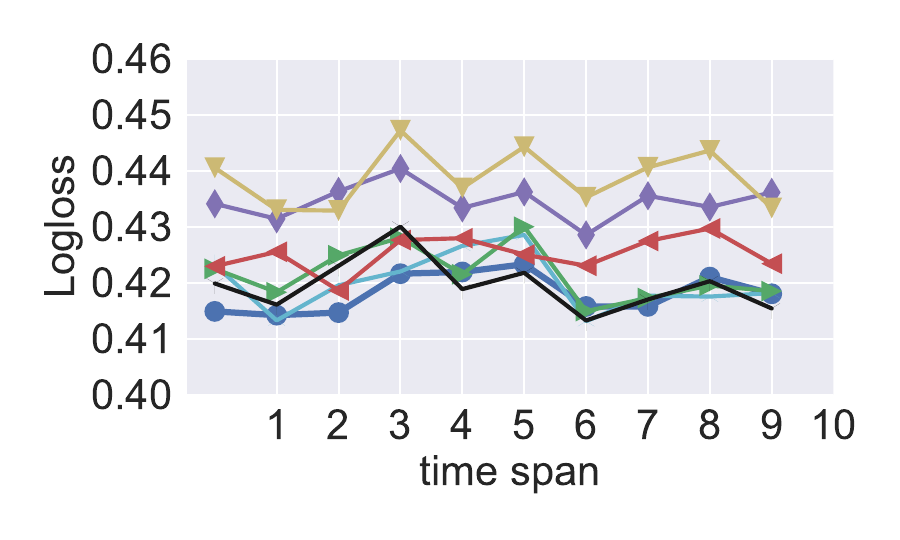}
	 \vspace{-0.26in}
	\caption{Logloss on Avazu}
	\end{subfigure}
	
	\begin{subfigure}[b]{0.245\textwidth}
     \centering
	\includegraphics[width=\textwidth]{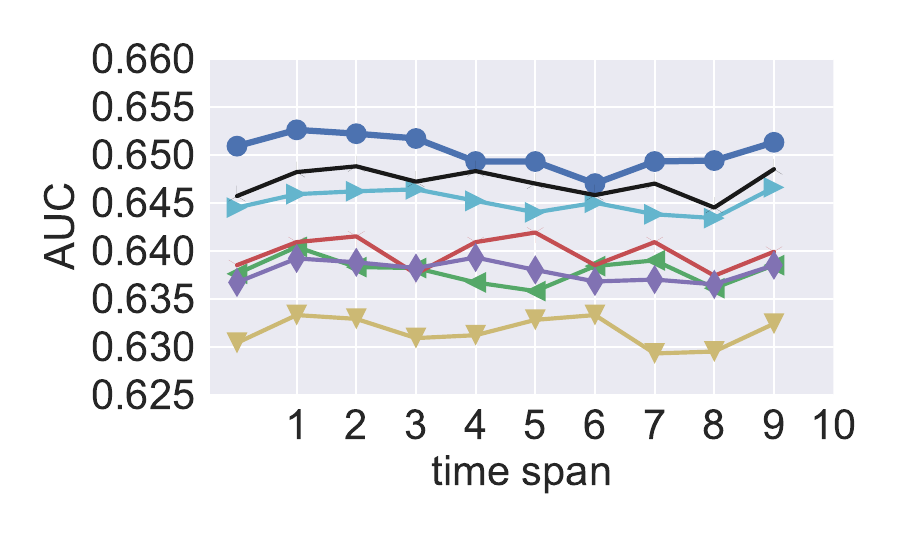}
		 \vspace{-0.26in}
	\caption{AUC on Media}
	\end{subfigure}
	\hspace{-4mm}
	\begin{subfigure}[b]{0.245\textwidth}
     \centering
	\includegraphics[width=\textwidth]{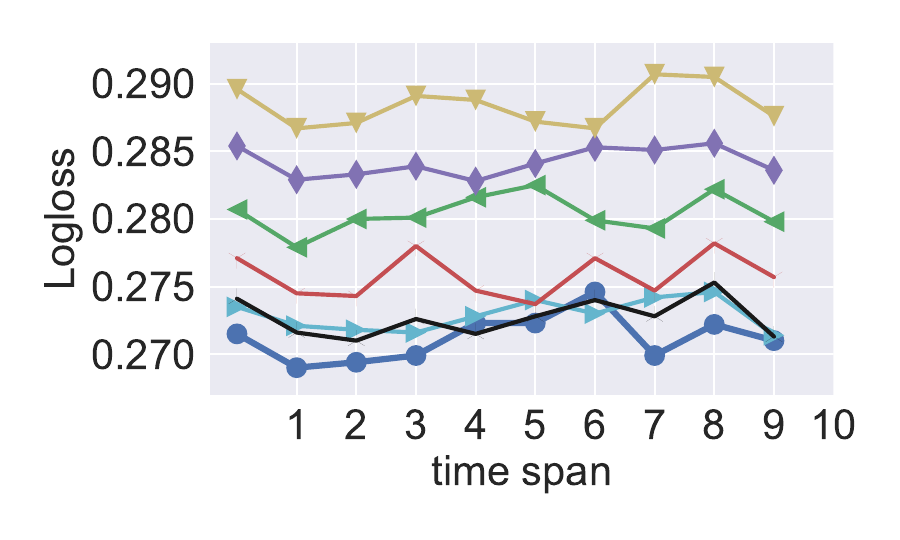}
	 \vspace{-0.26in}
	\caption{Logloss on Media}
	\end{subfigure}

	\vspace{-0.1in}
	\caption{Prediction performance on each time span. We present the six most competitive baselines and omit the remaining methods.}
	\label{result}
\end{figure}

 First, RS performs slightly better than IU, which reflects that the future samples also need the old knowledge preserved in historical samples. Second, RMFX+SAR and GAG+SAR perform better than RMFX and GAG respectively, which shows restricting the feature embedding updating according to its staleness can also improve the performance of methods with other sampling strategies. Third, EWC+SAS and ASMG+SAS perform better than EWC and ASMG respectively, which indicates using a fixed size reservoir to replay stale features is compatible with these model-based methods using other regularization strategies. Fourth, \model achieves 0.95\%, 1.38\%, 0.78\% and 1.72\% relative improvements compared to ASMG, the best baseline on four datasets, respectively. The reason why \model  achieves greater improvement on Media might be that Media has lower feature staleness extent
 which is easier to be overcome by SAS.
Figure~\ref{result} shows the disentangled performance at each time span. The performance of all methods has a similar trend on different datasets, and \model outperforms all the other methods in most time spans.

\subsection{Ablation Study: RQ2}

\begin{figure}
\vspace{-0.04in}
	\centering
	\begin{subfigure}[b]{0.49\textwidth}
     \centering
	\includegraphics[width=\textwidth]{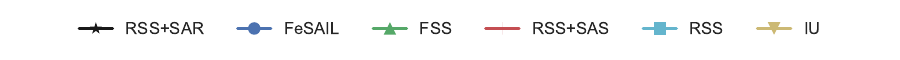}
\vspace{-0.25in}
	\end{subfigure}

	\begin{subfigure}[b]{0.45\textwidth}
     \centering
	\includegraphics[height=2.8cm,width=.96\textwidth]{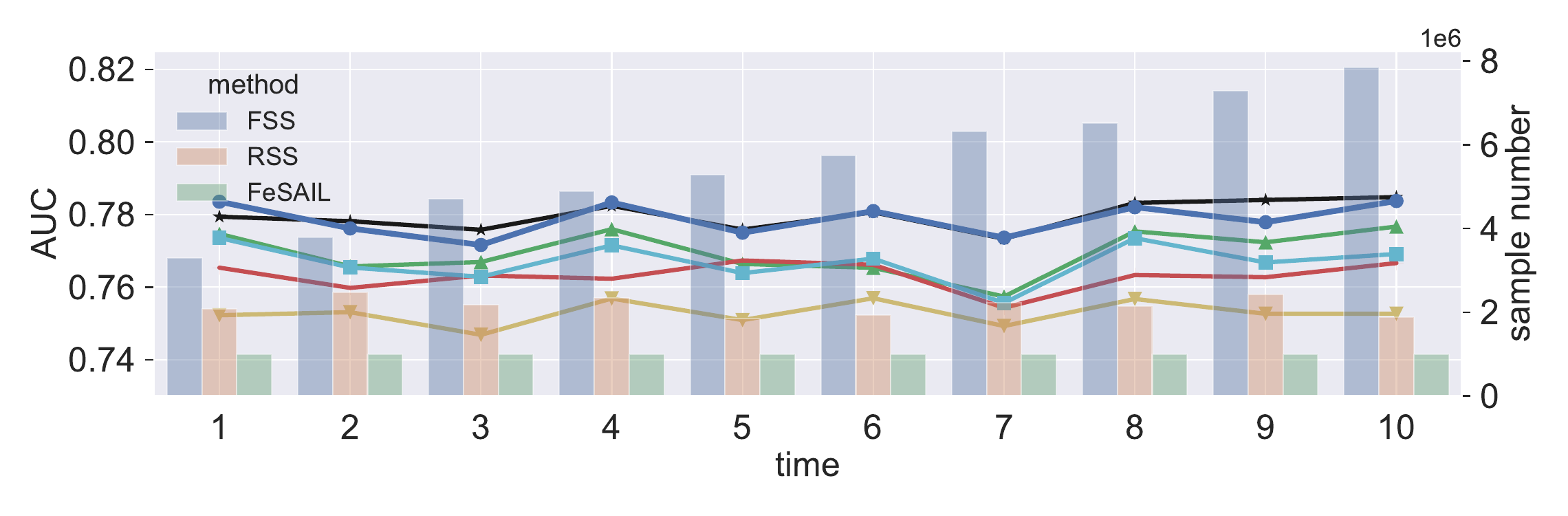}
	\vspace{-0.23in}
	\end{subfigure}
	
	\caption{Ablation study on Avazu.}
	\label{ablation}
	\vspace{-0.1in}
\end{figure}

\begin{table}
	\normalsize
	\centering
	\vspace{-0.0in}
	
	 \begin{tabular}{c|ccc}
	 \toprule
	Method/Runtime~(min) & sampling & training & total\\
	\midrule
	IU & 0 & 2,519 & 2,519 \\
	FSS & 190 & 2,519 & 2,709 \\
	RSS & 143 & 920 &1,063 \\
	RSS+SAR & 143 & 1,014 & 1,157 \\
	RSS+SAS & 401 & 920 & 1,321 \\
	\textbf{\modelns} & 401 & 1,014 & 1,415 \\
	\bottomrule
	\end{tabular}
	\caption{Average sampling and retraining time on Avazu.}
	\label{speed}
	\vspace{-0.1in}
\end{table}

We perform an ablation study on iPinYou and Avazu to evaluate the effects of different components of \model on the performance. Similar conclusions can be drawn on Criteo and Media~(results can be found in supplementary material). Specifically, we consider the following settings.
\begin{itemize}[itemsep=0pt,topsep=0pt,parsep=0pt,leftmargin=10pt]
 \item \textbf{Incremental Update (IU)\footnote{For a fair comparison, we supplement the latest historical samples to make the total sample size the same as FSS.}} updates the model incrementally using only the new data.
  \item \textbf{Full Stale Sampling~(FSS)} updates the model with new data and all the historical samples containing stale features.
  \item \textbf{Reservoir-based Stale Sampling~(RSS)} updates the model using new data and all the samples in the reservoir.
   \item \textbf{RSS+SAR} add the SAR on \textbf{RSS}.
   \item \textbf{RSS+SAS} restrict the reservoir size of \textbf{RSS} using SAS.
    \item \textbf{\modelns} using both the SAR and SAS on \textbf{RSS}.
   \end{itemize}
The results are presented in Figure~\ref{ablation}. The line chart reflects the AUC difference between different methods. We can see that \model performs best among all the comparison methods, which shows the removal of any component from \model will hurt the final performance. SAR and SAS jointly contribute to improving the \textbf{RSS} performance.
The bar chart reflects the sample number in each time span. We can see that \textbf{FSS} samples size increases during incremental learning due to the growing size of historical samples. However, the sample sizes of \textbf{RSS} and \textbf{\modelns} maintain relatively stable. Moreover, the sample size of \textbf{\modelns} is fixed and controllable, which is more friendly to real-world incremental recommendation systems.
One concern of SAS is to control the time cost of sampling and model training. Table~\ref{speed} shows the average sampling time and the retraining time of each method over different time spans on Avazu. 
The time cost for SAR is negligible compared to the total training time cost. SAS can remarkably reduce the sample size and accelerate model retraining. In the meanwhile, thanks to the neighbour-based SAS algorithm, the sampling time cost for SAS is bearable.

\subsection{Parameter Sensitivity: RQ3}
\begin{figure}
\vspace{-0.1in}
	\centering
	
	\begin{subfigure}[b]{0.49\textwidth}
     \centering
	\includegraphics[width=\textwidth]{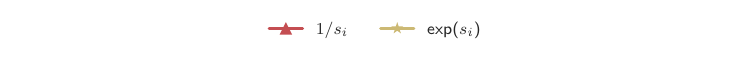}
	\vspace{-0.25in}
	\end{subfigure}
	
	\begin{subfigure}[b]{0.38\textwidth}
     \centering
	\includegraphics[width=\textwidth]{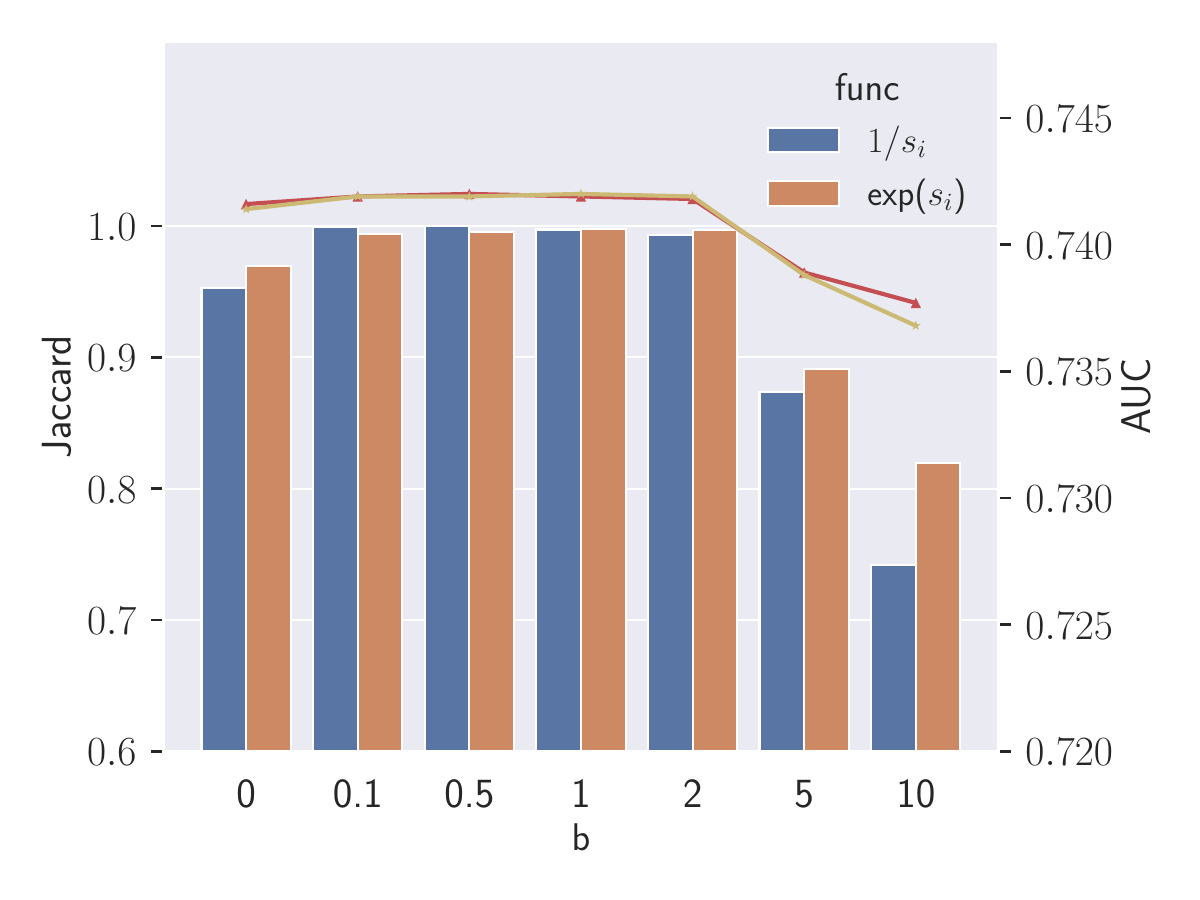}
	 \vspace{-0.34in}
	\end{subfigure}
	
	\caption{Parameter sensitivity \textit{w.r.t} different inverse correlation functions and biases on Avazu.}
	\label{parameter}
\end{figure}

We investigate the sensitivity of choosing different inverse correlation functions $func$ and bias $b$ in Eq.(\ref{func}). We choose iPinYou and Avazu, and report the average AUC~(line chart) and Jaccard Similarity~(JS) of chosen samples set with the control group~(the group with the highest AUC) over ten time spans. The $func$ is chosen between inverse proportional function $1/s_i$ and negative exponential function $\exp(s_i)$. The $b$ is chosen from $[0, 0.1, 0.5, 1, 2, 5, 10]$, where the larger $b$ means the less consideration for feature staleness. As shown in Figure~\ref{parameter}, the performance on two datasets achieve a similar JS with two $func$ and $b$ in the range $[0.1, 0.5, 1, 2]$, which means the chosen samples set is insensitive to the $func$ and the $b$ within a moderate range and productively achieve a competitive AUC towards the control group. The reason might be that only the relative weight sum order of samples will influence the choosing order in Algorithm~\ref{algsas}, and choosing different $func$ and $b$ with small values will not change the relative weight sum order of samples.

\subsection{Case Study: RQ4}
\label{sec:casestudy}

\begin{figure}

\vspace{-0.08in}
	\centering
	\includegraphics[width=0.48\textwidth]{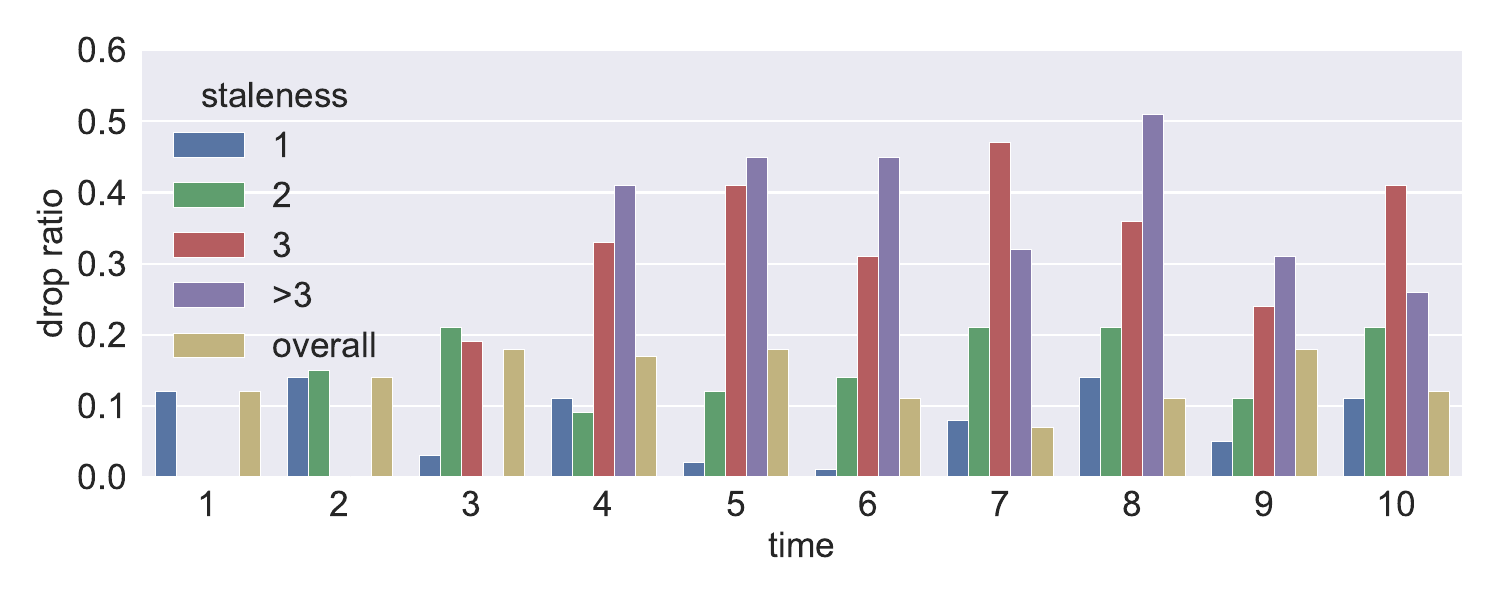}
	
		\vspace{-0.18in}
	\caption{The drop ratios of feature groups with different staleness on iPinYou.}
	\label{casestudy:ipinyou}
		
\end{figure}

%
%
%

We perform a case study on iPinYou to illustrate the effectiveness of SAS in selecting samples. As presented in Figure~\ref{casestudy:ipinyou}, we have two observations.
First, the overall feature drop ratio after SAS is around 10\%, which is consistent with the approximation ratio given by Algorithm~\ref{algsas}.
Second, the features with larger staleness have a higher drop ratio in most time spans except in the first several time spans, which do not have features with staleness larger than its time span index. It means SAS prefers to choose features with smaller staleness, which is consistent with our motivation and mitigates the low-frequency features' influence on performance.

%
%

\section{Conclusion}
In this paper, we show the feature staleness problem existing in CTR prediction 
and propose a feature staleness aware IL method for CTR prediction~(FeSAIL) to mitigate this problem by adaptively replaying samples containing stale features.
We first introduce a feature staleness aware sampling algorithm~(SAS) to guarantee sampling efficiency. We then introduce a feature staleness aware regularization mechanism~(SAR) to restrict the low-frequency feature updating. 
We conduct extensive experiments and demonstrate that FeSAIL can efficiently tackle the feature staleness problem. On average, it achieves 1.21\% AUC improvement compared to various state-of-the-art methods on three widely-used public datasets and a private dataset. 

\section*{Acknowledgements}
This work is supported by the National Key Research and Development Program of China (2022YFE0200500), Shanghai Municipal Science and Technology Major Project (2021SHZDZX0102), the Tencent Wechat Rhino-Bird Focused Research Program, and SJTU Global Strategic Partnership Fund (2021SJTU-HKUST). Yanyan Shen is the corresponding author.

\bibliographystyle{named}
\bibliography{ijcai23}

\begin{thebibliography}{}

\bibitem[\protect\citeauthoryear{Barrett \bgroup \em et al.\egroup }{2017}]{reservior5}
Rick Barrett, Rick Cummings, and Eugene Agichtein.
\newblock {Streaming Recommender Systems}.
\newblock In {\em IJCAI}, pages 381--389, 2017.

\bibitem[\protect\citeauthoryear{Caverlee \bgroup \em et al.\egroup }{2020}]{tisasrec}
James Caverlee, Xia~"Ben" Hu, Mounia Lalmas, Wei Wang, Jiacheng Li, Yujie Wang, and Julian McAuley.
\newblock {Time Interval Aware Self-Attention for Sequential Recommendation}.
\newblock {\em WSDM}, pages 322--330, 2020.

\bibitem[\protect\citeauthoryear{Chekuri and Kumar}{2004}]{mcp}
Chandra Chekuri and Amit Kumar.
\newblock Maximum coverage problem with group budget constraints and applications.
\newblock In {\em Approximation, Randomization, and Combinatorial Optimization. Algorithms and Techniques}, pages 72--83. Springer, 2004.

\bibitem[\protect\citeauthoryear{Chen \bgroup \em et al.\egroup }{2013}]{reservior6}
Chen Chen, Hongzhi Yin, Junjie Yao, and Bin Cui.
\newblock {TeRec: a temporal recommender system over tweet stream}.
\newblock In {\em RecSys}, volume~6, pages 1254--1257, 2013.

\bibitem[\protect\citeauthoryear{Cheng \bgroup \em et al.\egroup }{2020}]{cheng}
Weiyu Cheng, Yanyan Shen, and Linpeng Huang.
\newblock Adaptive factorization network: Learning adaptive-order feature interactions.
\newblock In {\em AAAI}, volume~34, pages 3609--3616, 2020.

\bibitem[\protect\citeauthoryear{d'Aquin \bgroup \em et al.\egroup }{2020}]{tien}
Mathieu d'Aquin, Stefan Dietze, Claudia Hauff, Edward Curry, Philippe~Cudre Mauroux, Xiang Li, Chao Wang, Bin Tong, Jiwei Tan, Xiaoyi Zeng, and Tao Zhuang.
\newblock {Deep Time-Aware Item Evolution Network for Click-Through Rate Prediction}.
\newblock In {\em Proceedings of the 29th ACM International Conference on Information \& Knowledge Management}, pages 785--794, 2020.

\bibitem[\protect\citeauthoryear{Diaz-Aviles \bgroup \em et al.\egroup }{2012}]{rmfx}
Ernesto Diaz-Aviles, Lucas Drumond, Lars Schmidt-Thieme, and Wolfgang Nejdl.
\newblock {Real-time top-n recommendation in social streams}.
\newblock In {\em WWW}, pages 59--66. 2012.

\bibitem[\protect\citeauthoryear{Diederik P.~Kingma}{2015}]{adam}
Jimmy Lei~Ba Diederik P.~Kingma.
\newblock {Adam: A Method for Stochastic Optimization}.
\newblock In {\em ICLR}, pages 785--794, 2015.

\bibitem[\protect\citeauthoryear{Guo \bgroup \em et al.\egroup }{2017}]{deepfm}
Huifeng Guo, Ruiming Tang, Yunming Ye, Zhenguo Li, and Xiuqiang He.
\newblock {DeepFM: A Factorization-Machine based Neural Network for CTR Prediction}.
\newblock {\em arXiv}, abs/1703.04247, 2017.

\bibitem[\protect\citeauthoryear{Guo \bgroup \em et al.\egroup }{2018}]{din}
Yike Guo, Faisal Farooq, Guorui Zhou, Xiaoqiang Zhu, Chenru Song, Ying Fan, Han Zhu, Xiao Ma, Yanghui Yan, Junqi Jin, Han Li, and Kun Gai.
\newblock {Deep Interest Network for Click-Through Rate Prediction}.
\newblock In {\em AAAI}, pages 1059--1068, 2018.

\bibitem[\protect\citeauthoryear{He and McAuley}{2016}]{fossil}
Ruining He and Julian McAuley.
\newblock {Fusing Similarity Models with Markov Chains for Sparse Sequential Recommendation}.
\newblock {\em arXiv}, 2016.

\bibitem[\protect\citeauthoryear{Huang \bgroup \em et al.\egroup }{2020a}]{sml}
Jimmy Huang, Yi~Chang, and Xueqi Cheng.
\newblock {How to Retrain Recommender System? A Sequential Meta-Learning Method}.
\newblock In {\em Proceedings of the 43rd International ACM SIGIR Conference on Research and Development in Information Retrieval}, pages 1479--1488, 2020.

\bibitem[\protect\citeauthoryear{Huang \bgroup \em et al.\egroup }{2020b}]{gag}
Jimmy Huang, Yi~Chang, and Xueqi~and Cheng.
\newblock {GAG: Global Attributed Graph Neural Network for Streaming Session-based Recommendation}.
\newblock In {\em SIGIR}, pages 669--678, 2020.

\bibitem[\protect\citeauthoryear{Kang and McAuley}{2018}]{sasrec}
Wang-Cheng Kang and Julian McAuley.
\newblock {Self-Attentive Sequential Recommendation}.
\newblock {\em ICDM}, pages 197--206, 2018.

\bibitem[\protect\citeauthoryear{Kirkpatrick \bgroup \em et al.\egroup }{2016}]{ewc}
James Kirkpatrick, Razvan Pascanu, Neil Rabinowitz, Joel Veness, Guillaume Desjardins, Andrei~A Rusu, Kieran Milan, John Quan, Tiago Ramalho, Agnieszka Grabska-Barwinska, Demis Hassabis, Claudia Clopath, Dharshan Kumaran, and Raia Hadsell.
\newblock {Overcoming catastrophic forgetting in neural networks}.
\newblock {\em arXiv}, 2016.

\bibitem[\protect\citeauthoryear{Krishnapuram \bgroup \em et al.\egroup }{2016}]{deepcrossing}
Balaji Krishnapuram, Mohak Shah, and JC~Mao.
\newblock {Deep Crossing: Web-Scale Modeling without Manually Crafted Combinatorial Features}.
\newblock In {\em KDD}, pages 255--262, 2016.

\bibitem[\protect\citeauthoryear{Lim \bgroup \em et al.\egroup }{2017}]{narm}
Ee-Peng Lim, Marianne Winslett, Mark Sanderson, Ada Fu, Jimeng Sun, and Shane Culpepper.
\newblock {Neural Attentive Session-based Recommendation}.
\newblock {\em CIKM}, pages 1419--1428, 2017.

\bibitem[\protect\citeauthoryear{Mi and Faltings}{2020}]{man}
Fei Mi and Boi Faltings.
\newblock {Memory Augmented Neural Model for Incremental Session-based Recommendation}.
\newblock In {\em Proceedings of the Twenty-Ninth International Joint Conference on Artificial Intelligence}, pages 2169--2176, 2020.

\bibitem[\protect\citeauthoryear{Mi \bgroup \em et al.\egroup }{2020}]{ader}
Fei Mi, Xiaoyu Lin, and Boi Faltings.
\newblock {ADER: Adaptively Distilled Exemplar Replay Towards Continual Learning for Session-based Recommendation}.
\newblock {\em arXiv}, 2020.

\bibitem[\protect\citeauthoryear{Peng \bgroup \em et al.\egroup }{2021}]{asmg}
Danni Peng, Sinno~Jialin Pan, Jie Zhang, and Anxiang Zeng.
\newblock {Learning an Adaptive Meta Model-Generator for Incrementally Updating Recommender Systems}.
\newblock In {\em SIGIR}, pages 411--421, 2021.

\bibitem[\protect\citeauthoryear{Qu \bgroup \em et al.\egroup }{2016}]{ipinyou1}
Yanru Qu, Han Cai, Kan Ren, Weinan Zhang, Yong Yu, Ying Wen, and Jun Wang.
\newblock Product-based neural networks for user response prediction.
\newblock In {\em 2016 IEEE 16th International Conference on Data Mining (ICDM)}, pages 1149--1154. IEEE, 2016.

\bibitem[\protect\citeauthoryear{Rendle \bgroup \em et al.\egroup }{2010}]{fpmc}
Steffen Rendle, Christoph Freudenthaler, and Lars Schmidt-Thieme.
\newblock {Factorizing personalized Markov chains for next-basket recommendation}.
\newblock {\em WWW}, pages 811--820, 2010.

\bibitem[\protect\citeauthoryear{Sun \bgroup \em et al.\egroup }{2019}]{bert4rec}
Fei Sun, Jun Liu, Jian Wu, Changhua Pei, Xiao Lin, Wenwu Ou, and Peng Jiang.
\newblock {BERT4Rec: Sequential Recommendation with Bidirectional Encoder Representations from Transformer}.
\newblock In {\em Proceedings of the 28th ACM International Conference on Information \& Knowledge Management}, 2019.

\bibitem[\protect\citeauthoryear{Sun \bgroup \em et al.\egroup }{2022}]{sun2022self}
Li~Sun, Junda Ye, Hao Peng, Feiyang Wang, and Philip~S Yu.
\newblock Self-supervised continual graph learning in adaptive riemannian spaces.
\newblock In {\em IJCAI}, 2022.

\bibitem[\protect\citeauthoryear{Teredesai \bgroup \em et al.\egroup }{2019}]{ssr}
Ankur Teredesai, Vipin Kumar, and Ying Li.
\newblock {Streaming Session-based Recommendation}.
\newblock In {\em IJCAI}, pages 1569--1577, 2019.

\bibitem[\protect\citeauthoryear{Wang \bgroup \em et al.\egroup }{2017}]{dcn}
Ruoxi Wang, Bin Fu, Gang Fu, and Mingliang Wang.
\newblock {Deep \& Cross Network for Ad Click Predictions}.
\newblock In {\em ADKDD}, pages 1--7, 2017.

\bibitem[\protect\citeauthoryear{Wang \bgroup \em et al.\egroup }{2020}]{incctr}
Yichao Wang, Huifeng Guo, Ruiming Tang, Zhirong Liu, and Xiuqiang He.
\newblock {A Practical Incremental Method to Train Deep CTR Models}.
\newblock {\em arXiv}, 2020.

\bibitem[\protect\citeauthoryear{Zhang \bgroup \em et al.\egroup }{2014}]{ipinyou}
Weinan Zhang, Shuai Yuan, Jun Wang, and Xuehua Shen.
\newblock {Real-Time Bidding Benchmarking with iPinYou Dataset}.
\newblock {\em arXiv}, 7 2014.

\bibitem[\protect\citeauthoryear{Zhikai and Yanyan}{2022}]{tamic}
Wang Zhikai and Shen Yanyan.
\newblock Time-aware multi-interest capsule network for sequential recommendation.
\newblock In {\em SDM}, pages 558--566, 2022.

\bibitem[\protect\citeauthoryear{Zhikai and Yanyan}{2023}]{imsr}
Wang Zhikai and Shen Yanyan.
\newblock Incremental learning for multi-interest sequential recommendation.
\newblock In {\em ICDE}, 2023.

\bibitem[\protect\citeauthoryear{Zhou \bgroup \em et al.\egroup }{2018}]{dien}
Guorui Zhou, Na~Mou, Ying Fan, Qi~Pi, Weijie Bian, Chang Zhou, Xiaoqiang Zhu, and Kun Gai.
\newblock {Deep Interest Evolution Network for Click-Through Rate Prediction}.
\newblock In {\em AAAI}, pages 5941--5948, 2018.

\bibitem[\protect\citeauthoryear{Zhu \bgroup \em et al.\egroup }{2019}]{fnn}
Wenwu Zhu, Dacheng Tao, Xueqi Cheng, Peng Cui, Elke Rundensteiner, David Carmel, Qi~He, Jeffrey~Xu Yu, Zekun Li, Zeyu Cui, Shu Wu, Xiaoyu Zhang, and Liang Wang.
\newblock {Fi-GNN: Modeling Feature Interactions via Graph Neural Networks for CTR Prediction}.
\newblock {\em Proceedings of the 28th ACM International Conference on Information and Knowledge Management}, pages 539--548, 11 2019.

\end{thebibliography}

\end{document}